\documentclass[aps,pra,english,   notitlepage,superscriptaddress,showpacs,nofootinbib]{revtex4-1}
\usepackage[a4paper, left = 3.02cm, right = 3.02cm, top = 2.85cm, bottom = 2.85cm]{geometry}
\usepackage[colorlinks=true, citecolor=blue, urlcolor=blue ]{hyperref}
\usepackage{graphicx}% Include figure files
\usepackage{dcolumn}% Align table columns on decimal point
\usepackage{bm}
\usepackage{amsmath,amsfonts}
\usepackage[english]{babel}
\usepackage{amsthm}
\usepackage{hyperref}
\usepackage{xcolor}
\theoremstyle{plain}

\theoremstyle{remark}

\newtheorem{theorem}{\indent Theorem}
\newtheorem{corollary}{\indent Corollary}
\newtheorem{lemma}{\indent Lemma}

\begin{document}
%\preprint{APS/123-QED}

%\title{Entanglement witnessed by arbitrarily many independent observers\\recycling a single quantum shared state not violating any Bell inequality %classically correlated
%}

\title{Sequentially witnessing entanglement by independent observer pairs}

% Force line breaks with \\
\author{Mao-Sheng Li}
\email{li.maosheng.math@gmail.com}
\affiliation{ School of Mathematics,
	South China University of Technology, Guangzhou
	510641,  China}

	\author{Yan-Ling Wang}
\email{wangylmath@yahoo.com}
\affiliation{ School of Computer Science and Technology, Dongguan University of Technology, Dongguan, 523808, China}
\date{\today}% It is always \today, today,

\begin{abstract}
	This study investigates measurement strategies in a scenario where multiple pairs of Alices and Bobs independently and sequentially observe entangled states. The aim is to maximize the number of observer pairs   $(A_k,B_l)$ that can witness entanglement. Prior research has demonstrated that arbitrary pairs $(A_k, B_k)$ ($k\leq n$)  can observe entanglement in all pure entangled states and a specific class of mixed entangled states [\href{https://doi.org/10.1103/PhysRevA.106.032419}{Phys. Rev. A \textbf{106} 032419 (2022)}]. However, it should be noted that other pairs $(A_k, B_l)$ with $(k\neq l \leq n)$ may not observe entanglement using the same strategy. Moreover, a novel strategy is presented, enabling every pair of arbitrarily many Alices and Bobs to witness entanglement   regardless of the initial state being a Bell state or a particular class of mixed entangled states.    These findings contribute to understanding measurement strategies for maximizing entanglement observation in various contexts.
\end{abstract}

\maketitle
\section{Introduction}
Quantum entanglement \cite{QE_HHHH09,QE_GT09} is of paramount importance in the field of quantum mechanics and finds numerous applications in   various quantum information processing tasks, including quantum cryptography \cite{QC_BB84,QC_Ekert91,QC_GRTZ02}, quantum key distribution\cite{QK_MY98,QK_BHK05,QK_ABGMPS07}, quantum metrology and sensing \cite{QM_PS09,QM_GLM11}, quantum simulation \cite{QS_GAN14},  and quantum computing \cite{QC_Preskill18}. Consequently, investigating efficient utilization of this entanglement resource is of great importance. One intriguing area of research in this field focuses on the recycling of entanglement that is shared among sequential observers, examining how effectively it can be reused.

 Bell nonlocality, as a powerful manifestation of entanglement, has attracted significant attention. In the scenario introduced by Silva et al. \cite{Silva15}, Alice and Bob share an entangled state. Bob measures his qubit and then passes it to a second Bob, who performs another measurement, and so on. The objective is to maximize the number of Bobs who can exhibit an expected violation of the Clauser-Horne-Shimony-Holt (CHSH) Bell inequality with a single Alice. 
 For bipartite quantum systems, it was shown in \cite{Silva15, Mal16} that, when each Bob performs sharp measurements with equal probability, at most two Bobs can achieve expected CHSH violations with a single Alice. However, by employing unsharp measurements, Brown and Colbeck \cite{Brown20} demonstrated that an arbitrarily large number of independent Bobs can share Bell nonlocality of the Bell state with a single Alice. Subsequently, these findings have been extended to higher dimensional bipartite systems \cite{Zhang21, Zhang22} and scenarios involving multiple Alices and Bobs \cite{Cheng21,Cheng22}. It has also been conjectured that,  in general, at most pair of Alices and Bobs can detect Bell-nonlocal correlations  \cite{Cheng21}. Furthermore, several studies have explored the recyclability of nonlocal correlations \cite{Mal16, Das19, Saha19, Maity20, Srivastava20,   Cabello21, Ren22, Zhu22, Das22, Hu22, Xi23}. These investigations contribute to our understanding of how nonlocal correlations can be effectively utilized in various settings, shedding light on the potential applications and implications of Bell nonlocality.

 Apart from studying entanglement through the violation of Bell inequalities, the use of entanglement witnesses (EWs) has proven to be valuable in the analysis and characterization of entanglement \cite{QE_HHHH09,QE_GT09,EW_HHH96,EW_T00}. An entanglement witness operator, denoted as  
 $W$, is an Hermitian operator that satisfies the following two conditions (1) $\langle W \rangle_{\rho_s} \geq 0$ for all separable states $\rho_s,$  (2)  there exists an entangled state $\rho_e$ such that $\langle W \rangle_{\rho_e} < 0$.  Here we use the notation $\langle W \rangle_\rho=\mathrm{Tr}[W\rho]$.

 Several works \cite{Bera18, Srivastava22,  Pandit22} have investigated whether entanglement can be sequentially witnessed by an arbitrary number of observers. In contrast to the case of Bell nonlocality, Pandit et al. \cite{Pandit22} demonstrated that entanglement can be detected and recycled an arbitrary number of times. Specifically, they considered a scenario where a bipartite entangled state is initially shared between the first Alice-Bob pair, denoted as $(A_1, B_1)$. To witness entanglement, the first pair, along with subsequent Alice-Bob pairs, performs local measurements on their respective qubits and passes them to the next pair $(A_2, B_2)$, and so on. The objective is to maximize the number of pairs that can independently witness entanglement. They found that an arbitrarily long sequence of Alice-Bob pairs can successfully detect entanglement sequentially. Recent works have extended similar findings to multipartite systems \cite{Srivastava22b,Srivastava22c}. However, it is important to note that the successful witnessing of entanglement by a pair $(A_k, B_k)$ does not imply the witness of entanglement for every pair $(A_i, B_j)$ with $i, j \leq k$. There are some pairs, as independent observers, may not witness entanglement themselves. Therefore, it is intriguing to explore whether entanglement can be detected by every pair of arbitrarily many Alices and Bobs. In this study, we will provide an affirmative answer to this problem.

The paper is structured as follows. In Section~\ref{sect:measurement}, we introduce the scenario under investigation and present the measurement strategies and entanglement witness employed in this study.
In Section~\ref{sect:main}, we demonstrate that every pair of arbitrarily many Alices and Bobs can sequentially witness the entanglement of any pure entangled states and a specific class of mixed entangled states. Finally, we provide a conclusion in Section~\ref{sect:conclusion}.

\section{Scenario,   measurement strategy and entanglement Witness }\label{sect:measurement}
Throughout this paper,  we denote $\{|0\rangle, |1\rangle\}$ as the computational basis of a qubit system. Under the computational basis, the  Pauli operators as $\sigma_i, i\in\{1,2,3\}$ are defined as follows
 $$
  \sigma_1=\left[
 \begin{array}{cc}
 	0 & 1 \\
 	1 &0 \\
 \end{array}
 \right],\ \ 
 \sigma_2=\left[
 \begin{array}{cc}
 	0& \mathrm{i} \\
 	-\mathrm{i} &0 \\
 \end{array}
 \right],\ \ 
 \sigma_3=\left[
 \begin{array}{cc}
 	1 & 0 \\
 	0 &-1 \\
 \end{array}
 \right],
 $$
  where $\mathrm{i}$ is the square root of  unit. And we will use the maximally entangled state $$|\Psi\rangle =\frac{1}{\sqrt{2}} (|01\rangle+|10\rangle)$$ 
  several times. Now we introduce the scenario that will consider and our measurement strategy.

\begin{figure}[t]
	\centering
	\includegraphics[width=1\columnwidth]{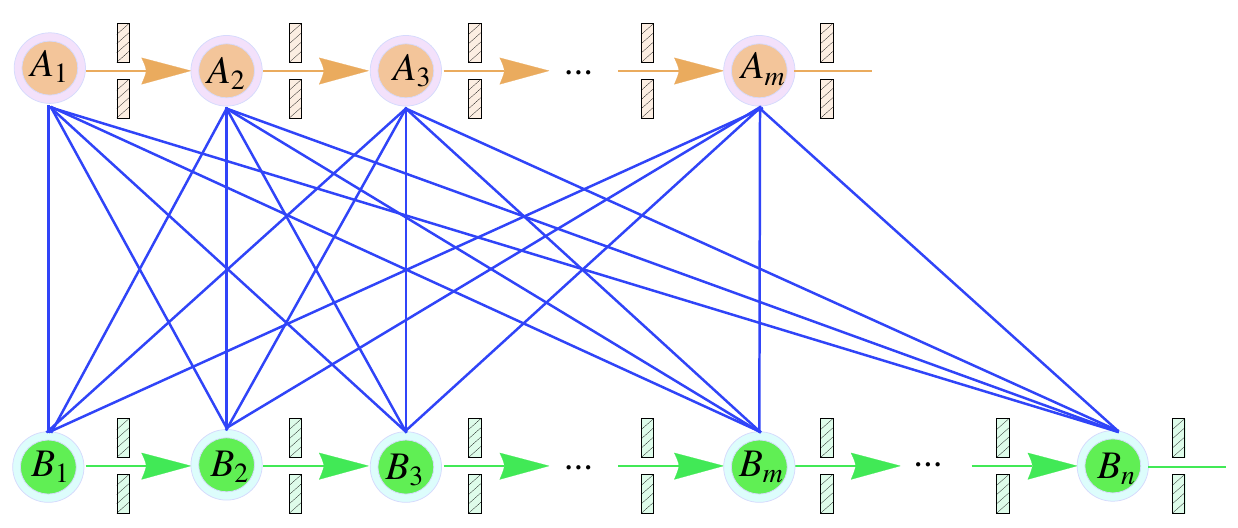}
	\label{plot}
	\caption{($(m,n)$-sequential scenario) Sequentially witnessing entanglement by multiple independent Alices and Bobs.  Each pair $(A_k,B_l)$ ($1\leq k\leq m, 1\leq l\leq n$) of Alices and Bobs are intended to detect the entanglement.}
	\label{fig.bipartitemn}
\end{figure}

 {\bf $(m,n)$-sequential scenario:} We investigate a generalized version of the sequential scenario described in \cite{Silva15}. Consider a scenario where there are $m$ sequential Alices and $n$ sequential Bobs, with an initial entangled state shared between the first Alice and Bob. The first Alice performs local measurements randomly on her qubit and then passes it to the next Alice, who performs subsequent measurements, and so on (see Fig.~\ref{fig.bipartitemn}). The same sequential process is followed by the Bobs.

The objective is to determine measurement strategies for both Alices and Bobs, such that each pair of Alices and Bobs can witness the entanglement. The ultimate goal is to maximize both $m$ and $n$ in order to successfully accomplish the aforementioned task. We will prove that by employing the weak measurement strategy presented below, one can achieve this task for arbitrarily long sequences of Alices and Bobs, regardless of whether the initial state is a pure entangled state or a specific class of mixed entangled states.

  {\bf Measurement strategy:}~  All the measurements we use are two-outcome positive operator-valued measurements (POVM) $\{M^{(0)},M^{(1)}\}$ where $M^{(i)}\geq 0$ and $M^{(0)}+M^{(1)}=\mathbb{I}_2$. For each $ P\in \{\sigma_1,\sigma_2,\sigma_3\}$ and $\lambda\in[ 0,1]$, we can define  a POVM $\{P^{(0)}(\lambda),P^{(1)}(\lambda)\}$ where
$$P^{(0)}(\lambda)=\frac{\mathbb{I}_2+\lambda P}{2},\ \ P^{(1)}(\lambda)=\frac{\mathbb{I}_2-\lambda P}{2}.$$ 
Suppose that each Alices and Bobs performs three types of POVMs $\{\sigma_i^{(0)}(\lambda),\sigma_i^{(1)}(\lambda)\}_{i=1}^3$  and $\{\sigma_j^{(0)}(\gamma),\sigma_j^{(1)}(\gamma)\}_{j=1}^3$   respectively with equal  probability. The $k$-th parameter of Alice is denoted as $\lambda^{(i)}_k$, and the $l$-th parameter of Bob is denoted as $\gamma^{(j)}_l$, where $i,j=1,2,3$. Moreover, throughout this paper, we assume that \begin{equation}\label{eq:lambda}
	\begin{array}{c}
		\lambda^{(3)}_k=1~~{\rm and}~~\lambda^{(1)}_k=\lambda^{(2)}_k=\lambda_k,\\[2mm]
		\gamma^{(3)}_l=1~~{\rm and}~~\gamma^{(1)}_l=\gamma^{(2)}_l=\gamma_l.
	\end{array}
\end{equation}
  Let $\rho_{A_k,B_l}$ denote the state  held by the $(A_k,B_l)$ pair. 
Now the pair $(A_k,B_l)$  try to calculate the total probability $p$ of the following three cases.

\begin{enumerate}
	\item [\rm(a)]  Both of Alice and Bob are taking measurements $\{\sigma_3^{(0)}(1),\sigma_3^{(1)}(1)\}$  and with the same outcomes (00 or 11).  The corresponding probability $p_1$ is  
	$$ \frac{1}{9}\mathrm{Tr} \left[\left(\frac{(\mathbb{I}_2+\sigma_3)\otimes (\mathbb{I}_2+\sigma_3)} {4}+ \frac{(\mathbb{I}_2-\sigma_3)\otimes (\mathbb{I}_2-\sigma_3)} {4} \right)\rho_{A_k,B_l}  \right]. $$ 
	 
	\item [\rm(b)] Alice and Bob are taking measurements $\{\sigma_1^{(0)}(\lambda_k),\sigma_1^{(1)}(\lambda_k)\}$ and $\{\sigma_1^{(0)}(\gamma_l),\sigma_1^{(1)}(\gamma_l)\}$ respectively but with different outcomes (01 or 10). The corresponding probability $p_2$ is 
	$$ \frac{1}{9}\mathrm{Tr} \left[\left(\frac{(\mathbb{I}_2+\lambda_k\sigma_1)\otimes (\mathbb{I}_2-\gamma_l\sigma_1)} {4}+\frac{(\mathbb{I}_2-\lambda_k\sigma_1)\otimes (\mathbb{I}_2+\gamma_l\sigma_1)} {4}\right)\rho_{A_k,B_l}\right].$$
	 
	\item [\rm(c)]  Alice and Bob are taking measurements $\{\sigma_2^{(0)}(\lambda_k),\sigma_2^{(1)}(\lambda_k)\}$ and $\{\sigma_2^{(0)}(\gamma_l),\sigma_2^{(1)}(\gamma_l)\}$ respectively but with different outcomes (01 or 10). The corresponding probability  $p_3$ is   
	$$ \frac{1}{9}\mathrm{Tr}\left[\left( \frac{(\mathbb{I}_2+\lambda_k\sigma_2)\otimes (\mathbb{I}_2-\gamma_l\sigma_2)} {4}+ \frac{(\mathbb{I}_2-\lambda_k\sigma_2)\otimes (\mathbb{I}_2+\gamma_l\sigma_2)} {4} \right) \rho_{A_k,B_l}\right].$$
	  
\end{enumerate}
So $p$ is equal to $p_1+p_2+p_3$ which can be simplified as 
\begin{equation}
	\begin{split} 
	 \frac{1}{9}\mathrm{Tr }\Big[\Big(\frac{3}{2}\mathbb{I}_4+\frac{1}{2}( \sigma_3\otimes \sigma_3- \lambda_k\gamma_l \sigma_1\otimes \sigma_1 
	 -\lambda_k\gamma_l \sigma_2\otimes \sigma_2)\rho_{A_k,B_l}\Big)\Big].
	\end{split} 
\end{equation}
Note that $\mathrm{Tr}[\mathbb{I}_4 \rho_{A_k,B_l}]=1$. Therefore, $18(p_1+p_2+p_3-\frac{1}{9})$ is equal to
$$  \mathrm{Tr }\left[( \mathbb{I}_4 + \sigma_3\otimes \sigma_3- \lambda_k\gamma_l \sigma_1\otimes \sigma_1-\lambda_k\gamma_l \sigma_2\otimes \sigma_2)\rho_{A_k,B_l} \right]. $$ 
Hence, we could define an observable
$W_{k,l}:=\mathbb{I}_4 + \sigma_3\otimes \sigma_3- \lambda_k\gamma_l \sigma_1\otimes \sigma_1-\lambda_k\gamma_l \sigma_2\otimes \sigma_2.$ Surprisingly, one could show that 
$$\langle W_{k,l}\rangle_\rho:=\mathrm{Tr}[W_{k,l} \rho]\geq 0$$ for all separable states when $0\leq \lambda_k,\gamma_l\leq 1$ (the proof is very similar to that in \cite{Srivastava22}. So we omit it here). So $W_{k,l}$ can be looked as  an entanglement witness when $0\leq \lambda_k,\gamma_l\leq 1$.  We wish to use this witness $W_{k,l}$ to observe the entanglement of $\rho_{A_k,B_l}.$ Therefore, we would like to find the condition of  $\lambda_k,\gamma_l $ such that 
\begin{equation}
	\begin{split} 
	\displaystyle\langle W_{k,l}\rangle_{\rho_{A_k,B_l}}=\mathrm{Tr } [( \mathbb{I}_4 + \sigma_3\otimes \sigma_3- \lambda_k\gamma_l \sigma_1\otimes \sigma_1 
	-\lambda_k\gamma_l \sigma_2\otimes \sigma_2)\rho_{A_k,B_l}  ]<0.
	\end{split} 
\end{equation}

Now we want to  express $ \langle W_{k,l}\rangle_{\rho_{A_k,B_l}}$ in term of the initial state $\rho_{A_1,B_1}.$ Note that $\rho_{A_k,B_l}$ can be  seem as a postmeasurement states of $\rho_{A_{k-1},B_l}$ or $\rho_{A_k,B_{l-1}}$.  By von Neumann-L\"uder's rule~\cite{Busch96}, we have     the following recursive relations

\begin{widetext}
\begin{eqnarray}\label{next}
	\rho_{A_kB_{l}}&=&\frac{1}{3}\sum_{i=1}^3\sum_{a=0}^1 \left(\sqrt{\sigma_i^{(a)}({\lambda^{(i)}_{k-1}})}\otimes \mathbb{I}_2\right) \rho_{A_{k-1}B_{l}} \left(\sqrt{\sigma_i^{(a)} ({\lambda^{(i)}_{k-1}})} \otimes \mathbb{I}_2\right) \nonumber, \\
	\rho_{A_kB_{l}}&=&\frac{1}{3}\sum_{j=1}^3\sum_{b=0}^1  \left(\mathbb{I}_2\otimes \sqrt{\sigma_j^{(b)}{(\gamma^{(j)}_{l-1}})} \right) \rho_{A_{k}B_{l-1}} \left( \mathbb{I}_2\otimes \sqrt{\sigma_j^{(b)}{(\gamma^{(j)}_{l-1}})}\right).    
\end{eqnarray}

By some calculation, Eqs.~\eqref{next} could be tranferred  to  
 \begin{eqnarray}\label{eq:State_recursive}
		\rho_{A_kB_{l}}&=&\frac{1}{6} \sum_{i=1}^3\Big[\left(1+\Lambda_{k-1}^{(i)}\right)   \rho_{A_{k-1}B_{l}}
		+ \left(1-\Lambda_{k-1}^{(i)}\right)  \sigma_i \otimes \mathbb{I}_2   \rho_{A_{k-1}B_{l}} \sigma_i \otimes \mathbb{I}_2  \Big],\nonumber\\
	\rho_{A_kB_{l}}&=&\frac{1}{6} \sum_{j=1}^3\Big[\left(1+\Gamma_{l-1}^{(j)}\right)   \rho_{A_{k}B_{l-1}}
	+ \left(1-\Gamma_{l-1}^{(j)}\right)   \mathbb{I}_2   \otimes  \sigma_j \rho_{A_{k}B_{l-1}} \mathbb{I}_2   \otimes  \sigma_j  \Big]
	\end{eqnarray}
where $\Lambda_{k}^{(i)}=\sqrt{1 - \lambda^{(i)~2}_{k}}$ and  $\Gamma_{l}^{(j)}=\sqrt{1 - \gamma^{(j)~2}_{l}}$. Therefore, $\Lambda_{k}^{(3)}= \Gamma_l^{(3)}=0$ and $\Lambda_k^{(i)}=\sqrt{1 - \lambda^{2}_{k}}$  and $\Gamma_{l}^{(j)}=\sqrt{1 - \gamma^{2}_{l}}$ for $i,j\in\{1,2\}$ which will be denoted as  $\Lambda_{k}$ and $\Gamma_l$ respectively for simplicity.
So the key three terms of $ \langle W_{k,l}\rangle_{\rho_{A_k,B_l}}$ could be expressed by
\begin{eqnarray}\label{eq:basic}
	&\text{Tr}[\sigma_1 \otimes \sigma_1 \rho_{A_kB_l}]=\text{Tr}[\sigma_1 \otimes \sigma_1 \rho_{A_1B_1}]\displaystyle\prod_{i=1}^{k-1}\frac{\left(1+\Lambda_i\right)}{3} \prod_{j=1}^{l-1}\frac{\left(1+\Gamma_j\right)}{3},\nonumber \\ 
		&\text{Tr}[\sigma_2 \otimes \sigma_2 \rho_{A_kB_l}]=\text{Tr}[\sigma_2 \otimes \sigma_2 \rho_{A_1B_1}]\displaystyle\prod_{i=1}^{k-1}\frac{\left(1+\Lambda_i\right)}{3} \prod_{j=1}^{l-1}\frac{\left(1+\Gamma_j\right)}{3}, \nonumber \\
		&\text{Tr}[\sigma_3 \otimes \sigma_3 \rho_{A_kB_l}]=\text{Tr}[\sigma_3 \otimes \sigma_3 \rho_{A_1B_1}]\displaystyle\prod_{i=1}^{k-1}\frac{\left(1+2\Lambda_i\right)}{3} \displaystyle\prod_{j=1}^{l-1}\frac{\left(1+2\Gamma_j\right)}{3}.
\end{eqnarray}

\end{widetext}

In Ref. \cite{Pandit22}, the authors presented that arbitrary many pairs of $(A_k,B_k)$ could detect the entanglement of  $\rho_{A_1B_1}=|\Psi\rangle\langle \Psi|$  sequentially . In fact,  the measurement strategies are  similar as above except that $\gamma_k=\lambda_k$.  They found that the pair $(A_k,B_k)$ will able to witness the entanglement if  and only if ${\rm Tr}[\rho_{A_kB_k}W_{k,k}]<0$, i.e.,
\begin{equation}
  \lambda_k^2>\displaystyle\frac{1-\prod_{l=1}^{k-1}\displaystyle\frac{\left(1+2\Lambda_{l}\right)^2}{9}}
	{2\prod_{l=1}^{k-1}\displaystyle\frac{\left(1+\Lambda_{l}\right)^2}{9}}.
\end{equation} 
Therefore, they  defined the sequence, $\lambda^2_k$, for $k \in \mathbb{N}$ as follows,
\begin{equation}\label{eq:sa}
	\lambda^2_k =
	\left\{
	\begin{array}{cc}
		(1+\epsilon) \displaystyle\frac{1-\displaystyle\prod_{l=1}^{k-1}\frac{\left(1+2\Lambda_{l}\right)^2}{9}}
		{2\displaystyle\prod_{l=1}^{k-1}\displaystyle\frac{\left(1+\Lambda_{l}\right)^2}{9}} &\mbox {if } \lambda^2_{k-1}\in (0,1) \\\\
		{\rm \infty}, &{\rm otherwise},
	\end{array}
	\right.
\end{equation}
with $0<\lambda^2_1<1$ and where $\epsilon > 0$.

Choosing $\lambda_1=0.005$ and $\epsilon=4$, the parameters $(\lambda_1,\lambda_2,\lambda_3,\lambda_4,\lambda_5)$  defined  by Eq.~\eqref{eq:sa} are
 $(0.0050, 0.0097, 0.0317, 0.1459, 0.9844).$
   Define the difference matrix  $D=(d_{kl})$  whose $(k,l)$ entry is 
   $$ d_{kl}:=\lambda_k \lambda_l-\displaystyle\frac{1-\displaystyle\prod_{i=1}^{k-1}\frac{\left(1+2\Lambda_{i}\right)}{3} \displaystyle\prod_{j=1}^{l-1}\frac{\left(1+2\Lambda_{j}\right)}{3}}
   {2\displaystyle\prod_{i=1}^{k-1}\frac{\left(1+\Lambda_{i}\right)}{3} \displaystyle\prod_{j=1}^{l-1}\frac{\left(1+\Lambda_{j}\right)}{3}}.$$
   One finds  that the pair $(A_k,B_l)$ will able to witness the entanglement if  and only if ${\rm Tr}[\rho_{A_kB_l}W_{k,l}]<0$, i.e., $d_{kl}>0.$ A numerical calculation shows that
$$ 100D\approx  \left[\begin{array}{rrrrr}
    0.0025  &  0.0042  &  0.0114   & 0.0099  & -1.4184\\[2mm]
0.0042 &    0.0075  &   0.0226 &    0.0446  &  -1.9159\\[2mm]
0.0114  &   0.0226  &   0.0802  &   0.3049 &   -1.2054\\[2mm]
0.0099  &   0.0446  &   0.3049  &   1.7031  &   7.5946\\[2mm]
-1.4184 &   -1.9159 &   -1.2054 &    7.5946  &  77.5252 
  \end{array}\right]
$$
   which indicates that although $A_5 B_5$ can detect the entanglement through $W_{55}$ but  $A_k B_l$ could not   through $W_{kl}$ for $(k,l)\in \{(1,5),(2,5),(3,5),(5,1),(5,2),(5,3)\}$ (see Fig. \ref{fig:five_pair} for an intuition). Therefore, it is interesting to change the measurement strategies such that all the corresponding pairs could detect the entanglement by themselves without communication with other members.

      \begin{figure}[h]
  	\centering
  	\includegraphics[width=0.75\columnwidth]{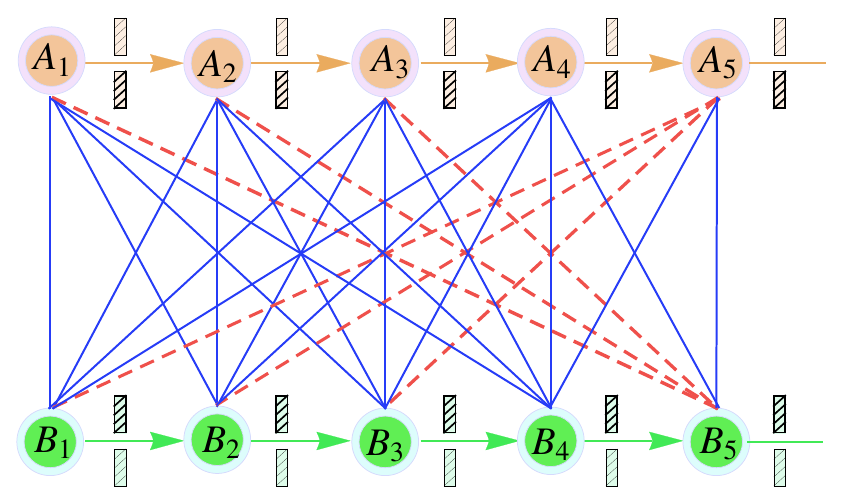}
  	\caption{Sequentially witnessing entanglement by multiple independent Alices and Bobs in Example.  The pair $(A_k,B_l)$ with blue line indicates the corresponding pair could witness the entanglement  while the  pair $(A_k,B_l)$ with red line indicates the corresponding pair could not witness the entanglement.  }
  	\label{fig:five_pair}
  \end{figure}

\section{Sequential entanglement detection between  arbitrarily many Alices and Bobs}\label{sect:main}

In this section, we aim at constructing some measurement strategies such that all the pairs $(A_k,B_l)$ ($1\leq k\leq m, 1\leq l\leq n$) could  detect the entanglement   regardless of the initial state being a Bell state or a particular class of mixed entangled states.  
\subsection{Maximally entangled state as initial state}
We assume that the initial state is $|\Psi\rangle$, i.e., it's density matrix is  $\rho_{A_1B_1}=|\Psi\rangle\langle \Psi|$. As $\{\mathbb{I}_2\otimes \mathbb{I}_2,\sigma_1\otimes \sigma_1,\sigma_2\otimes \sigma_2,\sigma_3\otimes \sigma_3\}$ form a basis of $4\times 4$ matrices,  it can be easily checked that 
\begin{equation}\label{eq:maxinitial}
\rho_{A_1B_1}=\frac{1}{4}\left(\mathbb{I}_2\otimes \mathbb{I}_2+\sigma_1\otimes \sigma_1+ \sigma_2\otimes \sigma_2-\sigma_3\otimes \sigma_3\right).
\end{equation}
As the detector is 
\begin{equation}\label{eq:witness}
	W_{k,l}= \mathbb{I}_4 +\sigma_3 \otimes \sigma_3 - \lambda_k\sigma_1 \otimes \gamma_l\sigma_1 - \lambda_k\sigma_2 \otimes \gamma_l\sigma_2,
\end{equation}
then by Eq.~\eqref{eq:basic}, we have  $\langle W_{k,l}\rangle_{\rho_{A_kB_l}}={\rm Tr}[W_{k,l}\rho_{A_kB_l}]$ which is equal to 
$$\begin{array}{l}
\displaystyle  1-\displaystyle\prod_{i=1}^{k-1}\frac{\left(1+2\Lambda_i\right)}{3} \displaystyle\prod_{j=1}^{l-1}\frac{\left(1+2\Gamma_j\right)}{3} 
-2\lambda_k\gamma_l\displaystyle\prod_{i=1}^{k-1}\frac{\left(1+\Lambda_i\right)}{3} \displaystyle\prod_{j=1}^{l-1}\frac{\left(1+\Gamma_j\right)}{3}.
\end{array} 
$$ 
Therefore, the first pair will observe entanglement if $-\frac{1}{2}\lambda_1\gamma_1={\rm Tr}[\rho_{A_1B_1}W_{1,1}]<0$, which implies
\begin{equation}
      \lambda_1\gamma_1>0.
\end{equation}
And the pair $A_kB_l$ will able to witness the entanglement if ${\rm Tr}[\rho_{A_kB_l}W_{k,l}]<0$, which implies
\begin{equation}\label{eq:total}
  \lambda_k \gamma_l>\displaystyle\frac{1-\displaystyle\prod_{i=1}^{k-1}\frac{\left(1+2\Lambda_{i}\right)}{3} \displaystyle\prod_{j=1}^{l-1}\frac{\left(1+2\Gamma_{j}\right)}{3}}
   {2\displaystyle\prod_{i=1}^{k-1}\frac{\left(1+\Lambda_{i}\right)}{3} \displaystyle\prod_{j=1}^{l-1}\frac{\left(1+\Gamma_{j}\right)}{3}}.
\end{equation} 
Let us define $\lambda_k(\theta)$, and $\gamma_l(\theta)$ for $k,l \in \mathbb{N}$ as   functions of a parameter $\theta$ on the interval $(0,1)$. Let $\lambda_1(\theta)=\gamma_1(\theta)=\theta$ and 
 for $k\geq 2$  
\begin{equation}\label{eq:define_lambdak}
\lambda_k(\theta) =
\left\{
\begin{array}{cc}
       \epsilon  \frac{1-\displaystyle\prod_{i=1}^{k-1}\frac{\left(1+2\Lambda_{i}(\theta)\right)}{3} }
      {L\gamma_1(\theta)\displaystyle\prod_{i=1}^{k-1}\frac{\left(1+\Lambda_{i}(\theta)\right)}{3} } &\mbox {if } \lambda_{k-1}(\theta)\in (0,1) \\\\
      {\rm \infty}, &{\rm otherwise},
\end{array}
\right.
\end{equation}
and 
\begin{equation}\label{eq:define_gammal}
\gamma_l(\theta) =\left\{
\begin{array}{cc}
	 \epsilon \frac{1-\displaystyle\prod_{j=1}^{l-1}\frac{\left(1+2\Gamma_{j}(\theta)\right)}{3} }	{L\lambda_1(\theta)\displaystyle\prod_{j=1}^{l-1}\frac{\left(1+\Gamma_{j}(\theta)\right)}{3} } &\mbox {if } \gamma_{l-1}(\theta)\in (0,1) \\\\
	{\rm \infty}, &{\rm otherwise},
\end{array}
\right.
\end{equation}
 where $\epsilon > 2$ and $L=2$ in this setting.
 
 \begin{lemma}\label{lemma:sequence_property}
 Let $L\leq 2<\epsilon$ and let $\big(\lambda_k(\theta)\big)_{k\in \mathbb{N}}$ 
 and$\big(\gamma_l(\theta)\big)_{l\in \mathbb{N}}$ be the sequences defined by Eqs.~\eqref{eq:define_lambdak} and \eqref{eq:define_gammal} respectively with $\lambda_1(\theta)=\gamma_1(\theta)=\theta.$ Then   the sequences  $\big(\lambda_k(\theta)\big)_{k\in \mathbb{N}}$ 
 and$\big(\gamma_l(\theta)\big)_{l\in \mathbb{N}}$ are the same sequence and  they are   increasing. Moreover, the limit $\lim_{\theta\rightarrow 0+} (\lambda_k(\theta)/\theta)$ exists  for any $k\in \mathbb{N}.$ Suppose that $a_k:=\lim_{\theta\rightarrow 0+} (\lambda_k(\theta)/\theta)$, then we have the following recursive relation
 $$ a_k=\frac{\epsilon}{3L}(\frac{3}{2})^{k-1}(\sum_{i=1}^{k-1} a_i^2).$$
 \end{lemma}

 \begin{theorem}\label{thm:Maximally}
  Let $L=2$ and  $ \epsilon\geq 4.$ 
    For any $m,n\in\mathbb{N},$  there exists a $\theta_{m,n}\in(0,1)$ such that the sequences $\big(\lambda_k:=\lambda_k(\theta_{m,n})\big)_{k=1}^m$ and  $\big(\gamma_l:=\gamma_l(\theta_{m,n})\big)_{l=1}^n$  satisfies inequalities \eqref{eq:total} for each pair $(k,l)$ ($1\leq k\leq m,$ $1\leq l\leq n$), that is, the detector $W_{k,l} $  defined by Eq.~\eqref{eq:witness} can witness the entanglement of $\rho_{A_kB_l}$    where  $\rho_{A_kB_l}$ is defined by the recursive relation \eqref{eq:State_recursive} with $\rho_{A_1B_1}$ given by Eq.~\eqref{eq:maxinitial}.
    
 \end{theorem}

 \begin{proof}
 By Lemma \ref{lemma:sequence_property}, 
 for each $k\in \mathbb{N},$  
 $\lambda_k(\theta)=\gamma_k(\theta)$ and  the limit $\lim_{\theta\rightarrow 0+} (\lambda_k(\theta)/\theta)$   exists. 
 Suppose $a_k:=\lim_{\theta\rightarrow 0+} (\lambda_k(\theta)/\theta)$ we also have
 $$ a_k=\frac{\epsilon}{3L}(\frac{3}{2})^{k-1}(\sum_{i=1}^{k-1} a_i^2)$$
 with $a_1=1.$ As $\epsilon\geq 4$, we can deduce that 
 $$ a_k\geq \sum_{i=1}^{k-1}a_i^2\geq a_{k-1} ^2\geq a_{k-1}\geq 1.$$ Set 
 $$ f_{kl}(\theta):=\displaystyle\frac{1-\displaystyle\prod_{i=1}^{k-1}\frac{\left(1+2\Lambda_{i}(\theta)\right)}{3} \displaystyle\prod_{j=1}^{l-1}\frac{\left(1+2\Gamma_{j}(\theta)\right)}{3}}
   {L\displaystyle\prod_{i=1}^{k-1}\frac{\left(1+\Lambda_{i}(\theta)\right)}{3} \displaystyle\prod_{j=1}^{l-1}\frac{\left(1+\Gamma_{j}(\theta)\right)}{3}}.$$
   With similar argument as the proof of Eq.~\eqref{eq:limitn1} in Appendix \ref{appendix:A}, one finds
   $$ \lim_{\theta\rightarrow 0+}  \frac{f_{kl}(\theta)}{\theta^2}=\frac{1}{3L}(\frac{3}{2})^{k+l-2}(\sum_{i=1}^{k-1} a_i^2+\sum_{j=1}^{l-1} a_j^2).$$
   Therefore, for $2\leq k\leq m, 2\leq l\leq n,$ the limit $L_{kl}:=\lim_{\theta\rightarrow 0+}  \frac{\lambda_k(\theta) \gamma_l(\theta)-f_{kl}(\theta)}{\theta^2}$ is equal to 
    $$ \frac{1}{3L}(\frac{3}{2})^{k+l-2}\left[ \frac{\epsilon^2}{3L} (\sum_{i=1}^{k-1} a_i^2)(\sum_{j=1}^{l-1} a_j^2)- (\sum_{i=1}^{k-1} a_i^2+\sum_{j=1}^{l-1} a_j^2)\right]$$
    which is positive provided the condition $\frac{\epsilon^2}{3L}>2.$  
    Let $\mathcal{M}=\max\{a_m,a_n\}$ and $\mathcal{L}=\min\{L_{k,l}\}_{k=2,l=2}^{m,n}.$ By the properties of limit, there exists a small enough $\theta_0>0$ (which is assumed to be less than $\frac{1}{\mathcal{M}+1}$ here)  such that for all $\theta_{m,n}\in(0,\theta_0)$ are satisfies
    $$ \begin{array}{l} \displaystyle  \lambda_i(\theta_{m,n})  = \gamma_i(\theta_{m,n}) <(\mathcal{M}+1)\theta_{m,n}<1,\\ [2mm]
   \displaystyle  \lambda_k(\theta_{m,n}) \gamma_l(\theta_{m,n})-f_{kl}(\theta_{m,n}) >\frac{\mathcal{L}}{2} {\theta_{m,n}^2}>0
     \end{array}
     $$
     for all $1\leq i\leq \max\{m,n\}, 2\leq k\leq m,2\leq l\leq n$ which implies the validity of inequalities of \eqref{eq:total} for these   $k,l.$
     On the other hand, for the pairs $(A_1,B_l)$ and $(A_k,B_1)$, the inequality \eqref{eq:total}    holds  by the definition of $\lambda_k(\theta_{m,n})$ and $\gamma_l(\theta_{m,n})$ as $\epsilon \geq4.$ 
     \qed
 \end{proof}

\begin{corollary} 
	For any positive integers $m,n$, if the first pair of Alices and Bobs in an  $(m,n)$-sequential scenario  share a  Bell state, then there exist some measurement strategies such that 
	all pairs $(A_k,B_l)$ ($1\leq k\leq m, 1\leq l\leq n$) could witness the entanglement sequentially. 
	
\end{corollary}
 
\subsection{Pure entangled states and a class of mixed entangled states as initial state}

Now we explore whether the aforementioned property holds true for other initial states.  It is well known that any two qubits entangled state could be locally unitary equivalent to the following form 
$$|\Psi_\alpha\rangle=\sqrt{\alpha}|01\rangle+\sqrt{1-\alpha}|10\rangle,$$
for some $\alpha\in (0,1]$. Now we assume that the initial state $\rho_{A_1,B_1}$ is of the following form 
\begin{equation}\label{eq:mixinitial}
	\rho_{A_1B_1}=p_1|\Psi_\alpha\rangle\langle\Psi_\alpha|+p_2|01\rangle\langle01|+p_3|10\rangle\langle10| 
\end{equation}
where $p_1>0$, $p_2,p_3\geq0$ and $p_1+p_2+p_3=1$.   Note that 
$$\begin{array}{l}
	 \mathrm{Tr}[\sigma_1\otimes \sigma_1 |01\rangle \langle 01|]=\mathrm{Tr}[\sigma_1\otimes \sigma_1 |10\rangle \langle 10|]=0,\\[2mm]
	 \mathrm{Tr}[\sigma_2\otimes \sigma_2 |01\rangle \langle 01|]=\mathrm{Tr}[\sigma_2\otimes \sigma_2 |10\rangle \langle 10|]=0,\\[2mm]
	 \mathrm{Tr}[\sigma_3\otimes \sigma_3 |01\rangle \langle 01|]=\mathrm{Tr}[\sigma_3\otimes \sigma_3 |10\rangle \langle 10|]=-1.
	\end{array}$$
As $
	W_{k,l}= \mathbb{I}_4 +\sigma_3 \otimes \sigma_3 - \lambda_k \gamma_l\sigma_1 \otimes\sigma_1 - \lambda_k\gamma_l\sigma_2 \otimes \sigma_2,$ we always have 
	$$\mathrm{Tr}[W_{k,l} |01\rangle\langle 01|]=\mathrm{Tr}[W_{k,l} |01\rangle\langle 01|]=0.$$
	Therefore,  
	$\mathrm{Tr}[W_{k,l}  \rho_{A_1,B_1}]= p_1 \mathrm{Tr}[W_{k,l}  |\Psi_\alpha\rangle\langle \Psi_\alpha|].$
Moreover,  $ |\Psi_\alpha\rangle\langle \Psi_\alpha|$ has similar decomposition  as of $|\Psi\rangle\langle\Psi|$ (see Eq. \eqref{eq:maxinitial})
\begin{equation}
   \begin{split}
    |\Psi_\alpha\rangle\langle\Psi_\alpha|= \frac{1}{4}\Big[\mathbb{I}_4  +2\sqrt{\alpha(1-\alpha)} \Big(\sigma_1 \otimes \sigma_1  
     +\sigma_2 \otimes \sigma_2 
    -\sigma_3 \otimes \sigma_3\Big)\Big].
\end{split} 
\end{equation}

   Therefore, the first pair will observe entanglement if $-p_1 4\sqrt{\alpha(1-\alpha)}\lambda_1\gamma_1={\rm Tr}[\rho_{A_1B_1}W_{1,1}]<0$, which implies
$
   	\lambda_1\gamma_1>0.
$
   And the pair $A_kB_l$ will able to witness the entanglement if ${\rm Tr}[\rho_{A_kB_l}W_{k,l}]<0$, which implies
   \begin{equation}\label{eq:totalmix}
   	\lambda_k \gamma_l>\frac{1-\displaystyle\prod_{i=1}^{k-1}\frac{\left(1+2\Lambda_{i}\right)}{3} \displaystyle\prod_{j=1}^{l-1}\frac{\left(1+2\Gamma_{j}\right)}{3}}
   	{4p_1\sqrt{\alpha(1-\alpha)}\displaystyle\prod_{i=1}^{k-1}\frac{\left(1+\Lambda_{i}\right)}{3} \displaystyle\prod_{j=1}^{l-1}\frac{\left(1+\Gamma_{j}\right)}{3}}.
   \end{equation}

Let us define $\lambda_k(\theta)$, and $\gamma_l(\theta)$ for $k,l \in \mathbb{N}$ as   functions of a parameter $\theta$ on the interval $(0,1)$. Let $\lambda_1(\theta)=\gamma_1(\theta)=\theta$ and 
 for $k\geq 2$  we define $\lambda_k(\theta) $ and $ \gamma_l(\theta)$ by  using the same recursive relations  as  Eqs.~\eqref{eq:define_lambdak} and \eqref{eq:define_gammal}  but  we assume $L=4p_1\sqrt{\alpha(1-\alpha)}$ here. We also have every possible pairs $(A_k,B_l)$ of arbitrarily many Alices and Bobs can observe the entanglement  by the sequentially witness defined    in the above setting.

 \begin{theorem}\label{thm:Mix}
      Let $L=4p_1\sqrt{\alpha(1-\alpha)}$ and $ \epsilon\geq 4.$ For any $m,n\in\mathbb{N},$  there exists a $\theta_{m,n}\in(0,1)$ such that the sequences $\big(\lambda_k:=\lambda_k(\theta_{m,n})\big)_{k=1}^m$ and  $\big(\gamma_l:=\gamma_l(\theta_{m,n})\big)_{l=1}^n$  satisfies inequalities \eqref{eq:total} for each pair $(k,l)$ ($1\leq k\leq m,$ $1\leq l\leq n$), that is, the detector $W_{k,l} $  defined by Eq.~\eqref{eq:witness} can witness the entanglement of $\rho_{A_kB_l}$  where  $\rho_{A_kB_l}$ is defined by the recursive relation \eqref{eq:State_recursive} with $\rho_{A_1B_1}$ given by Eq.~\eqref{eq:mixinitial}.
    
 \end{theorem}
 The proof is very similar to that of Theorem \ref{thm:Maximally} once one notice that $L\leq 2$ and $\epsilon^2/3L>2.$ So we omit the proof here.

\begin{corollary} 
	For any positive integers $m,n$, if the first pair of Alices and Bobs in an  $(m,n)$-sequential scenario  share a mix entangled state   given by Eq. \eqref{eq:mixinitial}, then there exist some measurement strategies such that 
	all pairs $(A_k,B_l)$ ($1\leq k\leq m, 1\leq l\leq n$) could witness the entanglement sequentially. 
	
\end{corollary}
\vskip 20pt

\section{Conclusion}\label{sect:conclusion}
In this study, we investigated the capability of sequential and independent observers to witness entanglement. It was shown that an arbitrary number of pairs of Alices and Bobs $(A_k,B_k)$ could witness entanglement sequentially.  We first shown that there are some cases that   $(A_k,B_l)$ could not witness the entanglement even that $(A_K,B_K)$ ($K\geq \max\{k,l\}$) could witness. Furthermore, we demonstrated the feasibility of a measurement strategy that enables all possible pairs $(A_k, B_l)$ of arbitrarily many Alices and Bobs to observe entanglement in both pure entangled states and a specific class of mixed entangled states. Although it has been conjectured that at most a single pair of Alices and Bobs can detect Bell-nonlocal correlations in the context of CHSH Bell nonlocality~\cite{Cheng21, Zhu22}, our work shows that quantum entanglement is comparatively easier to observe than Bell nonlocality.
 
It has been discovered that an arbitrary number of triples consisting of Alices, Bobs, and Charlies, denoted as $(A_k, B_k, C_k)$, can sequentially witness the entanglement of Greenberger-Horne-Zeilinger (GHZ) states \cite{Srivastava22b}. This finding has also been extended to multipartite systems \cite{Srivastava22c}. Consequently, an interesting question arises: Is it possible for all triples $(A_k, B_l, C_m)$, comprising arbitrarily many Alices, Bobs, and Charlies, to sequentially witness entanglement? Additionally, exploring similar problems concerning multipartite systems adds further intrigue to the investigation.

\acknowledgements

This work is supported by National Natural Science Foundation of China  with Grant No. 12371458, the Guangdong Basic and Applied Basic Research Foundation under Grants No. 2023A1515012074, and the Science and
Technology Planning Project of Guangzhou under Grants No. 2023A04J1296.

\onecolumngrid
%\begin{center}
%{\textbf{\large Appedix}}
%\end{center}
\appendix

\renewcommand{\theequation}{A\arabic{equation}}
 
\setcounter{equation}{0}
\setcounter{table}{0}

 \section{Proof of Lemma \ref{lemma:sequence_property}}\label{appendix:A}
  If  $\lambda_1(\theta)=\gamma_1(\theta)=\theta,$  then  the sequences  $\big(\lambda_k(\theta)\big)_{k\in \mathbb{N}}$ 
 and$\big(\gamma_l(\theta)\big)_{l\in \mathbb{N}}$ satisfy the same recursive relation with the same initial condition. Therefore the two sequences must be the same.
One finds  $$\lambda_2(\theta)/\lambda_1(\theta) \geq  \epsilon \frac{1-\frac{1+2(1-\frac{\theta^2}{2})}{3}}{2\theta^2 \frac{1+\sqrt{1-\theta^2}}{3}}=\frac{ \epsilon}{2(1+\sqrt{1-\theta^2})}>\frac{ \epsilon}{4}\geq 1.$$ Now for $k\geq 2$ and  $\lambda_k(\theta) \in (0,1)$,  
\begin{equation}
    \frac{\lambda_{k+1}(\theta)}{\lambda_{k}(\theta)}=\frac{3}{(1+\Lambda_k(\theta))}\frac{1-\displaystyle\prod_{i=1}^{k}\frac{\left(1+2\Lambda_{i}(\theta)\right)}{3}}{1-\displaystyle\prod_{i=1}^{k-1}\frac{\left(1+2\Lambda_{i}(\theta)\right)}{3}} > \frac{3}{(1+\Lambda_k(\theta))}>1,
\end{equation}
since $\Lambda_k(\theta)=\sqrt{1-\lambda^2_k(\theta)} \in (0,1)$.  So the   sequence $\big(\lambda_k(\theta)\big)_{k\in \mathbb{N}}$ is     increasing.  The next important observation about this sequence is that $\lim_{\theta\rightarrow 0+} \lambda _k(\theta)/\theta $ exists for all $k \in \mathbb{N}$.  Now we prove the statement by induction on $k$. For $k=1$, this  statement holds as $\lambda_1(\theta)/\theta\equiv 1$ by assumption. Now we suppose that $\lim_{\theta\rightarrow 0+} \frac{\lambda_k(\theta)}{\theta}=a_k\geq 1$ for $1\leq k\leq n.$ In the following, we will show 
$$\lim_{\theta\rightarrow 0+}  \frac{1-\displaystyle\prod_{k=1}^{n }\frac{\left(1+2\Lambda_{k}(\theta)\right)}{3} }
{ \theta^2}= \frac{1}{3}\sum_{k=1}^n a_k^2.$$
With this, we could show
\begin{equation}\label{eq:limitn1}
\lim_{\theta\rightarrow 0+}   \frac{\lambda_{n+1}(\theta)}{\theta} =\lim_{\theta\rightarrow 0+}\epsilon  \frac{1-\displaystyle\prod_{k=1}^{n }\frac{\left(1+2\Lambda_{k}(\theta)\right)}{3} }
{L\theta^2\displaystyle\prod_{k=1}^{n }\frac{\left(1+\Lambda_{k}(\theta)\right)}{3} }=\frac{\epsilon}{L}(\frac{3}{2})^n \lim_{\theta\rightarrow 0+}  \frac{1-\displaystyle\prod_{k=1}^{n }\frac{\left(1+2\Lambda_{k}(\theta)\right)}{3} }
{ \theta^2}=\frac{\epsilon}{3L}(\frac{3}{2})^n\sum_{k=1}^n a_k^2.
\end{equation}
Note that  $\lim_{\theta\rightarrow 0+} \Lambda_k(\theta)=\lim_{\theta\rightarrow 0+} \sqrt{1-\lambda_k(\theta)^2}=1.$ Moreover,  by the Talor expansion $\sqrt{1-x}=1-\frac{x}{2} +o(x)$, we have 
$$\lim_{\theta\rightarrow 0+} \frac{\Lambda_k(\theta)-1}{\theta^2}=\lim_{\theta\rightarrow 0+}   \frac{1-\frac{\lambda_k(\theta)^2}{2}+o(\lambda_k(\theta)^2)-1}{\theta^2}=-\frac{a^2_k}{2}.$$ Define $F_k(\theta)=\frac{2}{3}(\Lambda_k(\theta)-1).$ Then $\frac{ 1+2\Lambda_{k}(\theta) }{3}=1+F_k(\theta).$  So  $$\displaystyle\prod_{k=1}^{n }\frac{\left(1+2\Lambda_{k}(\theta)\right)}{3}=\displaystyle\prod_{k=1}^{n } (1+F_k(\theta))=1+\sum_{k=1}^n F_k(\theta)+\sum_{1\leq k<l\leq n}  F_{k}(\theta)F_l(\theta)+\cdots+F_{1}(\theta)\cdots F_n(\theta).$$ If we denote $H_n(\theta)=\sum_{1\leq k<l\leq n}  F_{k}(\theta)F_l(\theta)+\cdots+F_{1}(\theta)\cdots F_n(\theta)$, clearly, $\lim_{\theta\rightarrow 0+} \frac{H_n(\theta)}{\theta^2}=0.$ Therefore, we have 
$$\lim_{\theta\rightarrow 0+}  \frac{1-\displaystyle\prod_{k=1}^{n }\frac{\left(1+2\Lambda_{k}(\theta)\right)}{3} }
{ \theta^2}=\lim_{\theta\rightarrow 0+}  \frac{1-\displaystyle\prod_{k=1}^{n }(1+F_k(\theta)) }
{ \theta^2}=-\lim_{\theta\rightarrow 0+} \left(\sum_{k=1}^n \frac{F_k(\theta)}{\theta^2}+ \frac{H_n(\theta)}{\theta^2} \right)=\frac{1}{3}\sum_{k=1}^n a_k^2.$$

 %\twocolumngrid
%\bibliography{main}

\end{document}